\documentclass[a4paper,twocolumn,11pt,accepted=2024-04-08]{quantumarticle}
\pdfoutput=1
\usepackage[utf8]{inputenc}
\usepackage[english]{babel}
\usepackage[T1]{fontenc}
\usepackage{amsmath}

\usepackage{graphicx, color, graphpap}
\usepackage{float}
\usepackage{enumitem}
\usepackage{amssymb}
\usepackage{amsthm}
\usepackage{bbold}
\usepackage{booktabs}

\usepackage[ruled]{algorithm2e}
\usepackage{algorithmic}
\usepackage{mathrsfs}
\usepackage{bbm}
\usepackage{bm}
\usepackage{braket}
\usepackage{mathtools}
\usepackage[bookmarks=false,colorlinks=true,urlcolor=blue,citecolor=blue,linkcolor=blue]{hyperref}

\usepackage[sort&compress, numbers]{natbib}
\newtheorem{theorem}{Theorem}
\newtheorem{corollary}{Corollary}

\newtheorem{definition}{Definition}
\newtheorem{claim}{Claim}

\begin{document}

\title{Partial Syndrome Measurement for Hypergraph Product Codes}

\author{Noah Berthusen}
\email{nfbert@umd.edu}
\author{Daniel Gottesman}
\affiliation{Joint Center for Quantum Information and Computer Science, NIST/University of Maryland, College Park, Maryland 20742, USA}


\begin{abstract}
Hypergraph product codes are a promising avenue to achieving fault-tolerant quantum computation with constant overhead. When embedding these and other constant-rate qLDPC codes into 2D, a significant number of nonlocal connections are required, posing difficulties for some quantum computing architectures. In this work, we introduce a fault-tolerance scheme that aims to alleviate the effects of implementing this nonlocality by measuring generators acting on spatially distant qubits less frequently than those which do not. We investigate the performance of a simplified version of this scheme, where the measured generators are randomly selected. When applied to hypergraph product codes and a modified small-set-flip decoding algorithm, we prove that for a sufficiently high percentage of generators being measured, a threshold still exists. We also find numerical evidence that the logical error rate is exponentially suppressed even when a large constant fraction of generators are not measured.
\end{abstract}

\maketitle


\section{Introduction}

Quantum computers have the theoretical potential to solve problems intractable for classical computers. However, realizing this potential requires dealing with the noise inherent in near- and far-term devices. One way of doing this is to redundantly encode the quantum information in a quantum error-correcting code (QECC) and manipulate the encoded states to do computation. The threshold theorem \cite{Aharonov_Ben-Or_1999, Kitaev_1997, Knill_Laflamme_Zurek_1998} guarantees that such a procedure can work for arbitrarily long circuits as long as the noise rate of the system is below some threshold. Although polylogarithmic overhead is needed in the general case, it was later shown that the use of asymptotically \textit{good} quantum low-density parity-check (qLDPC) codes could reduce the overhead to a constant \cite{Gottesman_2014}. The question of whether such codes existed was unanswered until recently \cite{breuckmann_2021, Panteleev_Kalachev_2022, Leverrier_Zemor_2022, Lin_Hsieh_2022}; however, these constructions are currently more theoretical than practical.

When implementing QECCs on hardware it is especially advantageous to use one that is qLDPC, as its stabilizer generators act on a constant number of qubits, and its qubits are involved in a constant number of stabilizer generators. For certain architectures, such as nuclear magnetic resonance or superconducting qubits, another desirable code property is \textit{locality}. A code is considered local in $\mathbb{Z}^2$ if, when embedded in a grid of size $\sqrt{n} \times \sqrt{n}$, its generators act on qubits within a ball of constant radius. Recently, a popular choice when implementing a code family with these properties has been the surface code and its variations \cite{bravyi1998quantum, Kitaev_2003}. While it has local, weight-four generators and a favorable $\Theta(\sqrt{n})$ distance scaling, the surface code has a rate, $k/n$, which tends to zero as $n$ approaches infinity. A qLDPC code family that avoids this issue is hypergraph product (HGP) codes \cite{6671468}. This construction has the same $\Theta(\sqrt{n})$ distance scaling, but now with a constant rate; the trade-off, however, is that the stabilizer generators of HPG codes are very nonlocal. It was shown in Refs.~\cite{Bravyi_Terhal_2009, Bravyi_Poulin_Terhal_2010} that there is an intimate relationship between locality and the corresponding code parameters. In particular, the distance $d$ for a local code in $\mathbb{Z}^2$ is bounded above by $O(\sqrt{n})$, and the number of logical qubits $k$ obeys the relation $kd^2 = O(n)$. As such, the surface code saturates these bounds. Later work \cite{Baspin_Krishna_2021, Baspin_Krishna_2021_2} more precisely quantified the amount of nonlocality required to surpass them.

\begin{figure*}
    \centering
    \includegraphics[width=0.8\textwidth]{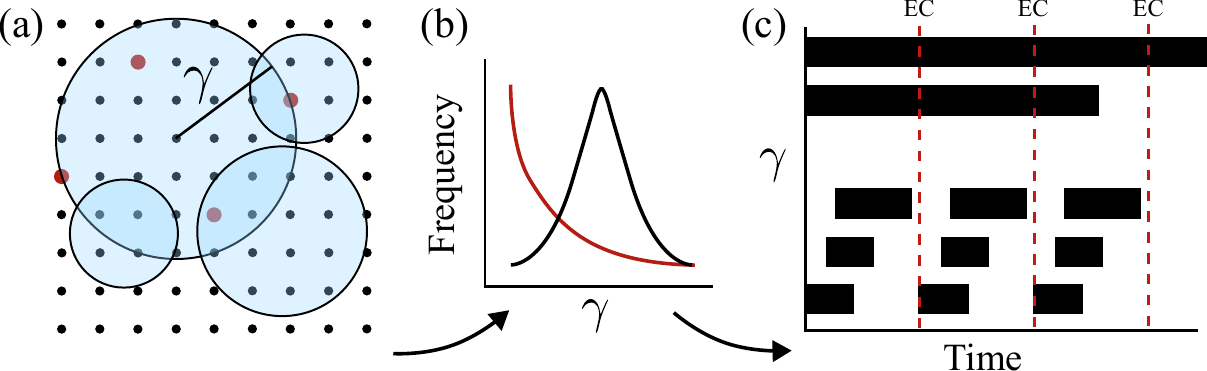}
    \caption{Overview of the stacked model. (a) After embedding a quantum code into $\mathbb{Z}^2$, each stabilizer generator has a parameter $\gamma$ that denotes the radius of the ball containing the qubits in its support. (b) Two possible distributions of $\gamma$ over the set of generators. The most advantageous distributions for this scheme are those where the relative frequency decays exponentially with increasing $\gamma$ (red curve). (c) An example \textit{schedule} for the generator measurements. The syndrome extraction circuits for the smaller generators are able to be prepared quickly, and so their syndromes are available during every round of error correction (red dashed lines). The larger generators require more time to build their syndrome extraction circuits, so this is done over a period of time that may stretch over several error correction rounds. More practically, priority is given to the smaller generators, and after completing them, the larger generators are worked on using any remaining time before an error correction round.  }
    \label{fig:stacked_model}
\end{figure*}

In this paper, we show through analytic and numerical evidence that repeated quantum error correction with HGP codes still provides a threshold even when a constant fraction of generators are measured only after many rounds of error correction.  This result suggests that it may be possible to build a fault-tolerant quantum computer with nonlocal qLDPC codes on architectures restricted to 2D local gates
with a procedure based on the \textit{stacked model} \cite{Baspin_Krishna_2021}. After embedding a QECC in a grid of size $\sqrt{n} \times \sqrt{n}$, the stabilizer generators are partitioned into a stack of layers based on the radius of the ball containing the qubits they act on. 
The bottom layer of the stack contains local generators, and as we move up the stack, the interaction radius increases while the number of generators of that size decreases.
Ideally, we can use codes which when embedded into $\mathbb{Z}^2$ have the property that the number of generators decreases exponentially with increasing radius; that is, a (large) constant fraction of the generators act on qubits within a support of constant radius. It was also shown in Ref.~\cite{Baspin_Krishna_2021} that any code constrained to the above model has a distance that is bounded by $\widetilde{O}(n^{2/3})$ and obeys the relation $k^3d^4 = \widetilde{O}(n^5)$. HGP codes satisfy this trade-off. \footnote{$\widetilde{O}(\cdot)$ is a variant of big $O$ notation that ignores log factors. $f(n) \in \widetilde{O}(h(n))$ is equivalent to $\exists k: f(n) \in O(h(n) \log^k n)$.}

The stacked model has a natural application when performing fault-tolerant quantum computations.
To convert a quantum circuit into a fault-tolerant version, the qubits are first encoded in some QECC, and then each operation in the original circuit is replaced with a fault-tolerant \textit{gadget}. Errors may still occur in the individual gates, so after each time step of the circuit, a round of fault-tolerant error correction is performed. 
To do this, the eigenvalues of the stabilizer generators of the code are measured to learn the syndrome, which is then used by a decoder to deduce and correct the error.
Measurement of the generators at the bottom of the stack takes constant time, since they are local. 
The corresponding syndrome information is then available during every round of error correction. 
As we move up the stack, the interaction radius increases. The important distinction to make is that while the generators are nonlocal, we are still measuring them with only 2D local gates, and so extracting these syndromes takes longer than for local generators. These nonlocal generators are measured less frequently than those lower in the stack, and their syndrome is not always available. This scheme is depicted in Fig.~\ref{fig:stacked_model}.

Several recent works have provided evidence against the possibility of doing error correction on architectures restricted to 2D local gates. Delfosse \textit{et al.}~\cite{Delfosse_Beverland_Tremblay_2021} investigated the problem of performing syndrome extraction circuits of HGP codes using 2D local gates and classical communication and presented numerical simulations suggesting that the resulting overhead was prohibitive. However, when considering the same problem in the context of the stacked model, it is possible that significant reductions in overhead could be gained by not measuring the generators with a large interaction radius every error correction cycle, since most of the work required is due to these very nonlocal generators.

We can roughly approximate the amount of work required to perform syndrome measurement using 2D local gates by estimating the number of SWAP gates in the extraction circuits. For a generator with interaction radius $\gamma$, the total number of SWAP gates needed to perform the syndrome measurement is proportional to $\gamma$.
As a concrete example, consider a qLDPC code on $n = 100,000$ qubits which when embedded into $\mathbb{Z}^2$ results in a generator distribution where the number of generators decays exponentially with increasing $\gamma$. 
Drawing $O(n)$ generators from this distribution and summing the radii of the smallest 90\% is $\sim3\%$ of the total sum across all generators.
Thus, we can estimate that the syndrome of these smallest 90\% of generators can be obtained using only $\sim3\%$ of the SWAP gates required to perform all of the syndrome measurements. 
Obtaining the remaining 10\% of the syndromes requires the majority of the work, but these circuits are built up over time (see Fig.~\ref{fig:stacked_model}(c)), allowing for a significant portion of the full error correction capabilities to be available during each error correction round. 
Although the resulting logical error rates will be strictly larger than when using a full syndrome, the reductions in overhead may outweigh the increases in the logical error rate.
A rigorous investigation of this question is the focus of further research~\cite{Berthusen_2023}.

Baspin \textit{et al.}~\cite{baspin2023lower} provide further evidence against 2D local implementations of qLDPC codes by deriving bounds on the amount of overhead needed to perform error correction at a given logical error rate. They show that the restriction to 2D local gates incurs polynomial overhead. However, they also note that their definition of error rate is very restrictive and that computations not satisfying this definition might not obey the overhead bound. It therefore remains possible that the stacked model could be used to perform these computations with constant overhead.
Apart from this brief discussion, we do not rigorously prove the feasibility of the stacked model as a whole or refute the claims put forth by these authors. This work only addresses the question of partial syndromes and their effect on performing error correction in the phenomenological noise model.

The remainder of the work is structured as follows. In Section~\ref{sec:background} we give a brief review of classical and quantum coding theory and introduce the families of codes relevant to this work. Section~\ref{sec:masking} introduces the idea of masking and contextualizes it with respect to the stacked model. In Section~\ref{sec:analytic_results}, we apply previous results to provide some analytical bounds on using masking during multi-round error correction. Section~\ref{sec:sims} provides empirical evidence to suggest that the analytical thresholds are better in practice. Finally, we conclude in Section~\ref{sec:discussion} with a summary and discussion of the remaining problems.

\section{Background}
\label{sec:background}
\subsection{Classical and Quantum Codes}

An $[n,k,d]$ binary linear code $\mathcal{C}$ encodes $k$ classical bits in a $k$-dimensional subspace of the $n$ bit, $n$-dimensional space, $\mathbb{F}_2^n$. Codewords are the binary vectors $v \in \mathbb{F}_2^n$ that satisfy the equation $H \cdot v=\textbf{0}$, where $H$ is a full rank binary matrix of size $(n-k) \times n$ called the \textit{parity check matrix}. The distance $d$ of a linear code is the minimum Hamming weight of a nonzero codeword. We can also represent the code $\mathcal{C}$ with its \textit{Tanner graph}, a bipartite graph $G = (V \sqcup C, E)$ whose biadjacency matrix is $H$.

An $[[n,k,d]]$ quantum error correcting code $\mathcal{Q}$ encodes $k$ logical qubits into a $2^k$-dimensional subspace of the $n$ qubit, $2^n$-dimensional Hilbert space, $(\mathbb{C}^2)^{\otimes n} = \mathbb{C}^{2^n}$. A commonly used class of QECCs are \textit{stabilizer codes} \cite{Gottesman_1997, Calderbank_Rains_Shor_Sloane_1997}. A stabilizer code is defined by its \textit{stabilizer} S, consisting of elements of the Pauli group
\begin{equation}
    \mathcal{P}_n = \{ I, X, Y, Z\}^{\otimes n} \times \{ \pm 1, \pm i\},
\end{equation}
whose action is the identity on the codewords of $\mathcal{Q}$. 
To have a codespace at all, we require that $-I \notin S$ and that S forms an abelian subgroup of $\mathcal{P}_n$. Denote by $N(S)$ the normalizer of S, the set of Paulis that commute with everything in the stabilizer, $N(S) = \{N \in \mathcal{P}_n \ | \ [N, M] = 0 \ \forall M \in S \}.$ The distance $d$ of $\mathcal{Q}$ is then defined to be the minimum weight of an operator in $N(S) \setminus S$. 
S is generated by $m = n-k$ independent \textit{stabilizer generators} $S = \langle S_1, ..., S_m \rangle$, which when measured provide an error syndrome of length $m$ used to deduce the error. We note that the syndrome labels the $2^m$ cosets of $P_n/N(S)$.

The \textit{binary symplectic representation} of a Pauli $P \in \mathcal{P}_n / \{\pm 1, \pm i\}$ is a bitstring consisting of two $n$-bit binary vectors, $(x|z) \in \mathbb{F}_2^{2n}$. The $i$th component of $x$ is 0 if $P$ acts on qubit $i$ with $I$ or $Z$ and 1 if $P$ acts on qubit $i$ with $X$ or $Y$. Similarly, the $i$th component of $z$ is 0 if $P$ acts on qubit $i$ with $I$ or $X$ and 1 if $P$ acts on qubit $i$ with $Z$ or $Y$. This transformation allows us to use techniques from classical coding theory on QECCs. In particular, we can represent the stabilizer generators as a $m \times 2n$ binary parity check matrix, $H$. If we consider then some error $E = (x|z) \in \mathbb{F}_2^{2n}$, the corresponding syndrome is $\sigma(E) = H\cdot E$, where multiplication and addition are performed over $\mathbb{F}_2$.

CSS codes \cite{Calderbank_Shor_1996} are a subclass of stabilizer codes where the stabilizer generators consist entirely of tensor products of $X$ and $I$ or $Z$ and $I$. As such, these codes have parity check matrices of the symplectic form $H = \big(\begin{smallmatrix} H_Z & 0 \\ 0 & H_X \end{smallmatrix}\big)$, with $H_Z \cdot H_X^T = H_X \cdot H_Z^T = \textbf{0}$ to enforce the abelian structure of S. In this form, it can be seen that decoding CSS codes can be broken down into decoding the two classical codes with parity check matrices $H_Z$ and $H_X$ separately, where $H_Z$ corrects bit-flip errors and $H_X$ corrects phase-flip errors. In this case, separate syndromes are needed to decode an error $E = (x|z)$,
\begin{equation}
    \sigma(E) = (\sigma_Z(x), \sigma_X(z)) = (H_Z \cdot x, H_X \cdot z).
\end{equation}
In this work, it may be unclear with respect to which stabilizer generators a syndrome is measured. Where clarification is needed, we slightly abuse notation and write a syndrome taken from a subset of the stabilizer $U \subseteq S$ as $\sigma_U(E)$. Using this notation, we do not explicitly specify the type of error we are measuring, but in all cases we will only consider one type. 
We let the Tanner graph of a CSS code, $G = (V \sqcup C_X \sqcup C_Z, E)$, to be the bipartite graph defined by its parity check matrix in symplectic form. The two disjoint sets of check nodes, $C_X, C_Z$, correspond to the $X-$ and $Z-$type stabilizer generators, respectively.

A classical or quantum code is considered a \textit{low density parity check} code if the weights of the rows and columns of its parity check matrix are bounded by a constant. Specifically, an $[[n,k,d]]$ stabilizer code is considered a $(\Delta_V, \Delta_C)-$qLDPC code if, for some constants $\Delta_V$ and $\Delta_C$, each qubit is involved in at most $\Delta_V$ stabilizer generators and each generator measures at most $\Delta_C$ qubits. 
We can equivalently say that the Tanner graph has bit node degree bounded by $\Delta_V$ and check node degree bounded by $\Delta_C$.

\subsection{Quantum Expander Codes}
\label{subsec:qexpander}

Hypergraph product (HGP) codes \cite{6671468} are CSS type codes made by taking the graph product of two classical codes $\mathcal{C}_1, \mathcal{C}_2$. When $\mathcal{C} := \mathcal{C}_1 = \mathcal{C}_2$ is a binary linear code with parameters $[n,k,d]$ and a full-rank parity check matrix, the parameters of the resulting hypergraph product code are $[[n^2 + (n-k)^2, k^2, d]]$. If the input code is $(\Delta_V, \Delta_C)-$LDPC  with $\Delta_V \le \Delta_C$, then the resulting quantum code is $(2\Delta_C, \Delta_V+\Delta_C)-$qLDPC. Furthermore, when the base code is replaced with a classical expander code \cite{Sipser_Spielman_1996}, the resulting quantum code is deemed a \textit{quantum expander code} and is equipped with a linear time decoding algorithm which we now describe.

The small-set flip (SSF) decoding algorithm \cite{Leverrier_Tillich_Zemor_2015} aims to imitate the classical flip decoding algorithm used to decode classical expander codes. Let $\mathcal{F}$ be the set of powersets of qubits in $X-$type generators and let $E$ be the initial $X-$type error. A single round takes as input a guessed error $\hat{E}_i$ and the syndrome of the remaining error $\sigma_i := \sigma_Z(E \oplus \hat{E}_i)$. The decoder then goes through all `small-sets' $f \in \mathcal{F}$ and finds the one that when flipped maximizes the decrease in syndrome weight, which is then applied to the guessed error for the next round. The algorithm succeeds if the final error has zero syndrome and is not a logical operation; otherwise, it fails. In other words, decoding is considered a success if the guessed error $\hat{E}$ is equivalent to the actual error $E$, that is $E \oplus \hat{E}$ belongs to the stabilizer group.
The success of the decoder is guaranteed for errors of size less than the distance, as well as random errors of linear size \cite{Fawzi_Grospellier_Leverrier_2018_2} provided the underlying classical codes are sufficiently expanding.

\begin{algorithm}[t]
\caption{Small-set flip decoding algorithm \cite{Leverrier_Tillich_Zemor_2015}}
\label{alg:ssf}
\begin{algorithmic}
\REQUIRE $(E, D)$
\STATE
\WHILE{$\exists F \in \mathcal{F} : |\sigma_i| - |\sigma_i \oplus \sigma_Z(F)| > 0$}
    \STATE $F_i = \max_{F\in\mathcal{F}} \frac{|\sigma_i|-|\sigma_i \oplus \sigma_Z(F)|}{|F|} $
    \STATE $\hat{E}_{i+1} = \hat{E}_i \oplus F_i$
    \STATE $\sigma_{i+1} = \sigma_i \oplus \sigma_Z(F_i)$
    \STATE $i = i+1$
\ENDWHILE
\STATE
\RETURN $\hat{E}_i$
\end{algorithmic}
\end{algorithm}

The complete decoding procedure is listed as pseudo-code in Algorithm~\ref{alg:ssf}. 
It takes as input a tuple $(E, D)$, where $E \subseteq V$ is an $X-$type error, and $D \subseteq C_Z$ is a potential syndrome error---that is the algorithm runs instead on the syndrome where some values have been flipped, $\sigma_Z(E) \oplus D$.
We make one small change to the algorithm for the purposes of using it in the context of the stacked model. Specifically, we exchange using the full syndrome for one taken from some subset of the stabilizer generators $U \subseteq S$. We still search through every opposite type generator when looking for small-sets $F$ to flip; however, the effect of flipping will only be visible on the restricted set of generators $\sigma_U(F)$. We overload the meaning of having the input $(E, D)$ when $D \subseteq C_Z$ is interpreted as a mask, in which case the available syndrome is $\sigma_D(F)$. The chosen interpretation will be clear from context.

 \subsection{Fault-Tolerance}
 \label{sec:masking_stacked}

A quantum circuit is considered fault-tolerant if it prevents errors from propagating throughout the circuit; in this way, it keeps the size of the error manageable for the QECC. 
We can convert a circuit into a fault-tolerant version by replacing each element of the original circuit with a fault-tolerant \textit{gadget} performing an equivalent operation on the encoded state. 
Fault-tolerant circuits can be naturally broken down into time steps, where a single time step consists of gadgets applied in parallel followed by error correction.

To investigate how an error propagates throughout a fault-tolerant circuit, we abstract the above model and instead work with the procedure described in Algorithm~\ref{alg:ft_alg}. For the purposes of analysis and simulation, we condense all gadgets, except error correction, into a single event that has an error with probability $p_\text{phys}$. We also assume that the error correction itself is ideal and that there is no syndrome error, except the artificially imposed error coming from the generators that have been masked with probability $p_\text{mask}$, which we now define. We later discuss how to make this scheme more realistic, but for the purposes of determining the effects of performing error correction with partial syndromes this simplified model is sufficient.

\begin{algorithm}[t]
  \caption{A simplified fault-tolerance scheme}
  \begin{algorithmic}
  \label{alg:ft_alg}
  \STATE Apply a mask $D$ with probability $p_\text{mask}$
  \STATE
  \FOR{t = 1, ..., $\tau$}
    \STATE Generate an error $F_t$ with probability $p_\text{phys}$ and apply $F_t$ to the current error:
    \begin{equation*}
        E_t' := F_t \oplus E_{t-1}
    \end{equation*}
    
    \STATE Run Algorithm~\ref{alg:ssf} on the input $(E_t, D)$ and correct using the decoded error $\hat{E_t}$:
    \begin{equation*}
        E_{t} := E_t' \oplus \hat{E_t}
    \end{equation*}
  \ENDFOR
  \STATE
  \STATE Generate an error $F_t$ with probability $p_\text{phys}$ and apply to the current error:
  \begin{equation*}
        E_\tau := F_\tau \oplus E_{\tau-1}
    \end{equation*}
  \STATE Run Algorithm~\ref{alg:ssf} on the input $(E_T, \varnothing)$
\end{algorithmic}
\end{algorithm}

\section{Syndrome masking}
\label{sec:masking}

The notion of \textit{masking} has recently been introduced as a way of describing fault-tolerant protocols for space-time codes \cite{Gottesman_2022}. We use the same idea here, although in a different context. An element of the stabilizer is considered masked if we cannot measure its eigenvalue during an error correction round. We follow the definition from \cite{Gottesman_2022} and define two subgroups of the stabilizer, $U$ and $T$, where $U \subseteq T \subseteq S$. The \textit{always unmasked} subgroup, $U$, are the stabilizers whose eigenvalues can be measured in a constant number of rounds, whereas the \textit{temporarily unmasked} subgroup, $T$, are the stabilizers whose eigenvalues can be measured in a number of rounds that can scale with the size of the code, $n$. 
In general, it could be the case that $T \subsetneq S$ where the set $S \setminus T$ contains stabilizers that cannot be measured on any time scale. In this work, we consider the case $U \subset T = S$. 
The subgroups form valid stabilizer codes, and as such can be described by their parameters. Defining $k$ for these codes has no real meaning since logical information is not being stored in the subspace; however, we can define the corresponding distances, where 
\begin{equation}
    \label{eq:normal_dist_def}
    d_U = \min |N(U) \setminus S| \quad d_T = \min |N(T) \setminus S|.
\end{equation} 
In other words, $d_U \ (d_T)$ is the weight of the smallest Pauli operator outside of the full group that has zero syndrome when measuring only the stabilizer generators of $U \ (T)$. We call $d_U$ and $d_T$ masked distances, whereas $d$ is the unmasked distance. Note that $d_U \le d_T \le d$. 


Since not every generator is measured, the resulting syndrome may have less information about the error than would otherwise be available if the full set of stabilizer generators were measured. For any number of masked generators, there is a set of \textit{invisible} errors that have zero syndrome on the generators of $U$ (or $T$) while having a nonzero syndrome in $S$. In particular, the new normalizer $N(U)$ contains $N(S)$ as well as all cosets of $\mathcal{P}_n/N(S)$ labeled with undetectable error syndromes. 
Furthermore, errors that were previously correctable may no longer be uniquely identifiable with the syndrome of $U$ or $T$. Note that errors with a zero syndrome for $U$ do not immediately cause logical errors, unlike errors with a zero syndrome for all of $S$.  If an error has a nonzero syndrome for $T$, it will eventually be detected, once the generators of $T \setminus U$ are unmasked.  The risk is that such errors will accumulate over time and become logical errors before they can be corrected.

\subsection{Masking and the Stacked Model}

Identifying which layers of the stack are available during an error correction round corresponds to specifying the temporarily unmasked subgroup, $T_t$, at each time step in the circuit, $t = 1, ..., \tau$. The always unmasked subgroup, $U \subset T_t$, is static over the execution of the circuit and so can be specified at the beginning. This set contains all local generators, as their eigenvalues can be measured in constant time. 
$T_t$ will contain $U$ as well as any additional layers that have completed syndrome extraction between time $t-1$ and $t$.
Since, in general, we want to measure all generators throughout the course of the circuit, $\bigcup_t T_t = S$; however, it may not be the case that any one time step has all generators available. 

An equivalent interpretation is to specify $S \setminus T_t$, the set of generators whose eigenvalues are not available during time step $t$.
For the remainder of the work, we consider `applying' a mask $D$ to be specifying this set, $S \setminus T_t$.

\section{Analytic results}
\label{sec:analytic_results}

In this section, we consider previous results on HGP codes and the SSF decoder in the context of masking in a multi-round error correction procedure.
We consider qubit errors and syndrome masks that follow a local stochastic noise model.

\begin{definition}
    (Local stochastic error model). We say that an error $(E, D)$ is local stochastic if there are error parameters $(p_\text{phys}, p_\text{synd})$ such that for any $F$ and $L$, $\Pr[F \subseteq E, L \subseteq D] \le p_\text{phys}^{|F|}p_\text{synd}^{|L|}.$
\end{definition}
HGP codes in conjunction with the SSF decoder have several desirable properties that make them a strong contender for fault-tolerance with constant overhead. Most relevant to us is the fact that they can tolerate random qubit errors and syndrome errors of linear size, as stated in the following theorem.
\begin{theorem}
\label{thm:synd_error}
(modified from Fawzi, Grospellier, Leverrier \cite{Fawzi_Grospellier_Leverrier_2018}). There exists a non-zero constant $p_0 > 0$ such that the following holds. Suppose that the error $(E,D)$ each satisfy a local stochastic noise model with parameters $p_\text{phys}$ and $p_\text{synd}$ where $p_\text{phys} < p_0$ and $p_\text{synd} < p_0$. If we run Algorithm~\ref{alg:ssf} on the input $(E,D)$ then there exists a random variable $E_\text{ls} \subseteq V$ with a local stochastic distribution with parameter $p_\text{ls} := p_\text{synd}^{\Omega(1)}$ such that:
\begin{equation}
    \Pr\big[E_\text{ls} \ \text{and} \ E \oplus \hat{E} \ \text{are not equivalent}\big] \le e^{-\Theta(\sqrt{n})}
\end{equation}
\end{theorem}
In the analysis for the above theorem, Fawzi \textit{et al.} consider an error $D$ in the syndrome to be a subset of the stabilizer generators whose measurement results have been flipped. Very briefly, the argument requires that the syndrome error does not form clusters on the syndrome adjacency graph~\cite{Gottesman_2014} for it to be tolerable. As such, $p_0$ must be below the percolation threshold of the syndrome adjacency graph of $\mathcal{Q}$. This value is a constant that depends only on $\Delta_V$ and $\Delta_C$ of the code. We can turn the result of a masked measurement into the above form by randomly assigning measurement outcomes to the generators included in the mask. Thus, in this context, we can say that Theorem~\ref{thm:synd_error} holds when a mask---turned syndrome error---$D$ satisfies a local stochastic noise model with parameter $p_\text{mask} < p_0$. 

The above analysis is sufficient in the case where we do a single round of masked error correction; however, when we use the same mask over several rounds, we have to be more careful about accounting for the correlations between error sources. Following the notation of Algorithm~\ref{alg:ft_alg}, in each round $t$ we have the syndrome error from the mask $D$, any error that was not fully corrected in the previous round $E_{t-1}$, and a new error $F_t$. 
When considered individually, all three error sources are local stochastic described by parameters $p_{\text{mask}}, p_{\text{res}},$ and $p_{\text{phys}}$, respectively. When looked at together, the new error and the syndrome error are bounded by
\begin{equation}
    \Pr[F \subseteq E \ \text{and} \ L \subseteq D] \le p_{\text{phys}}^{|F|}p_{\text{mask}}^{|L|}
\end{equation}
and similarly for the residual error and the new error, as per the definition of a locally stochastic error. However, we would expect to see correlations arise between the residual error and the syndrome error over the rounds, and so together they no longer obey a local stochastic noise model. Instead, they are bounded by:
\begin{equation}
    \label{eq:new_error}
    \Pr[F \subseteq E \ \text{and} \ L \subseteq D] \le \min\{ p_{\text{res}}^{|F|}, \ p_{\text{mask}}^{|L|} \}
\end{equation}
When $\max\{p_{\text{res}}, p_{\text{mask}}\} < p_0$, the threshold from Theorem~\ref{thm:synd_error}, we can say that the probability of clustering is at most $e^{-\Theta(\sqrt{N})}$ by plugging the error bound in Eq.~\eqref{eq:new_error} into Theorem 17 (\cite{Fawzi_Grospellier_Leverrier_2018_2}). With this, we can apply Lemma 26 (\cite{Fawzi_Grospellier_Leverrier_2018}) to bound the probability of the residual error obeying a local stochastic distribution, $\Pr[S \subseteq E_\text{ls}]$. Besides the requirement that $E \cup D$ forms clusters with low probability, we need that $\Pr[L \subseteq D] \le p_\text{mask}^{|L|}$. Since we assumed that the mask was chosen according to a local stochastic error model, this statement is satisfied for all rounds $t \le \tau.$
We are then able to apply Theorem~\ref{thm:synd_error} in an iterative manner, yielding the following result.
\begin{theorem}
\label{thm:ft_synd_error}
(Grospellier \cite{Grospellier}). Let $p_0$ be the threshold of Theorem~\ref{thm:synd_error}, and let $p_\text{mask}$ and $p_{\text{phys}}$ be such that:
\begin{equation}
    p_\text{mask} < \Big( \frac{p_0}{2}\Big)^{\Omega(1)} \quad\text{and}\quad p_{\text{phys}} < \frac{p_0}{2}.
\end{equation}
Then Algorithm~\ref{alg:ft_alg} fails with probability at most $(\tau+1)e^{-\Theta(\sqrt{n})}.$
\end{theorem}
If the conditions for Theorem~\ref{thm:ft_synd_error} are satisfied, then we can make the failure probability for the procedure arbitrarily small by using larger codes. 
This result is perhaps surprising given the following two claims:

\begin{claim}
\label{thm:degree_distribution}
Applying a random mask $D$ with parameter $p_\text{mask}$ to a qLDPC code $\mathcal{Q}$ results in a code $\mathcal{Q}' = \mathcal{Q}(S\char`\\D)$ whose Tanner graph has the following degree distribution: 
\begin{equation}
    \Pr(\deg(\mathcal{Q}'_{|v}) = i) = { \deg (\mathcal{Q}_{|v}) \choose i } p_\text{mask}^{\deg (\mathcal{Q}_{|v})} (1 - p_\text{mask})^i
\end{equation}
\end{claim}
Here, we use the notation $\deg(Q_{|v})$ to mean the degree of node $v$ in $\mathcal{Q}$. Since we assume $\mathcal{Q}$ to be LDPC, $\deg(v)$ is bounded by a constant $\Delta_V$ for all $v$, but the values may differ between vertices.
\begin{corollary}
\label{cor:const_dist}
Randomly masking a constant fraction of generators results in a masked distance of $d_U = 1$ whp.
\end{corollary}

\begin{proof}
Applying Claim~\ref{thm:degree_distribution} with $p_\text{mask} = O(1)$ gives a degree distribution where $\Pr(\text{Degree of qubit } v = 0) = p_\text{mask}^{\deg(\mathcal{Q}_{|v})} = \Omega(1)$ for all qubits. In this case, an error on such a qubit has zero syndrome on the remaining unmasked generators, $U$. As this error would not be an element of the stabilizer, $d_U = 1$.
\end{proof}

Although the always unmasked subgroup $U$ has a bad distance $d_U$ with high probability, it is capable of performing enough error correction in the intermediate steps to prevent the accumulation of errors. When the full set of stabilizer generators are unmasked at the end of the multi-round decoding procedure, it is then likely able to correct any residual errors. As we will now show, we see similar behavior at masking percentages well above what is guaranteed analytically. 

\section{Numerical simulations}
\label{sec:sims}

\begin{figure}
    \centering
    \includegraphics[width=\linewidth]{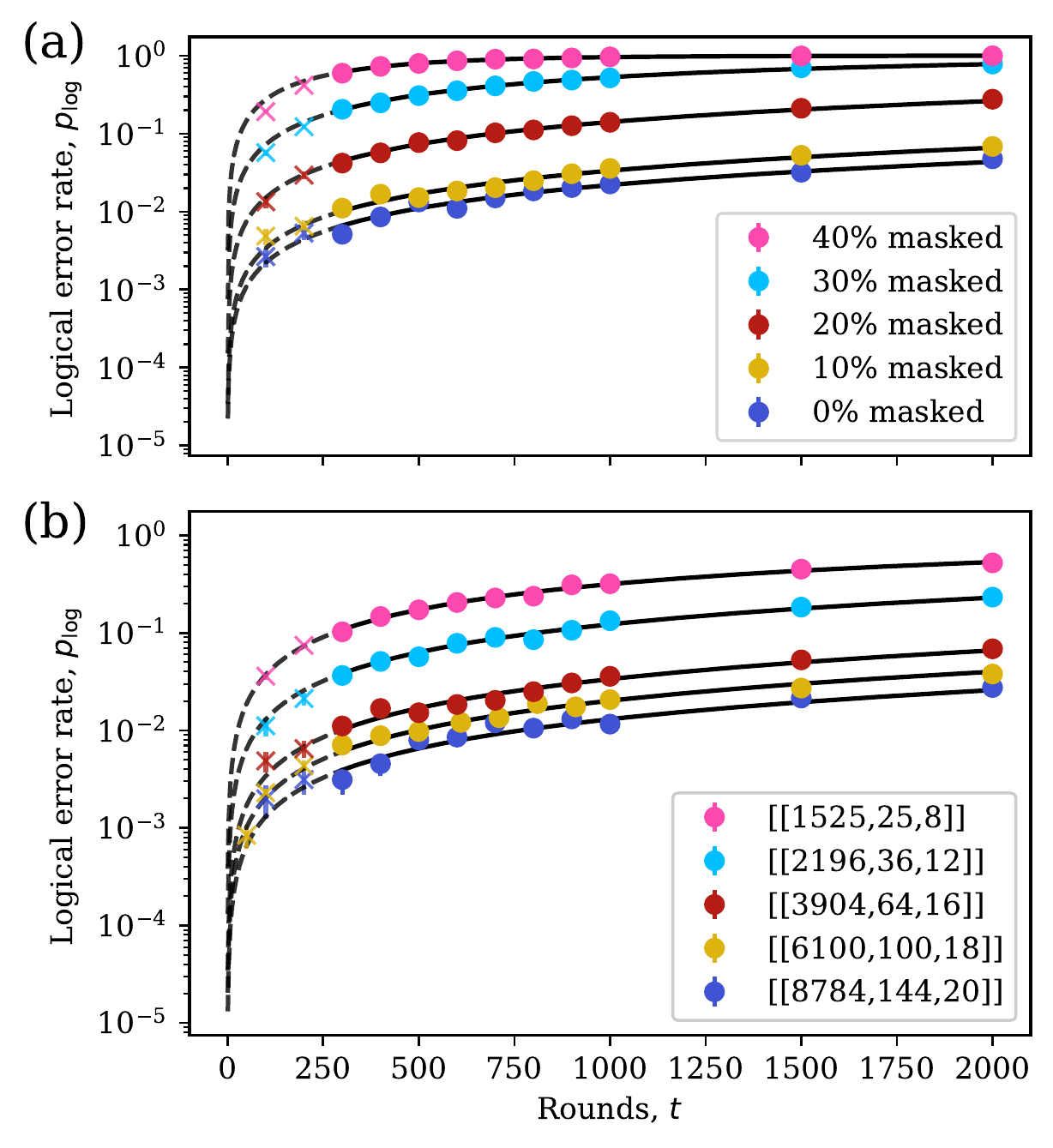}
    \caption{(a) Semilog plot of logical error rate as a function of the number of rounds for a $[[3904,64,16]]$ code and the simple unmasking schedule. (b) Logical error rate as a function of rounds across the $(12,11)$-qLDPC code family with fixed $p_\text{mask} = 10\%$ and the simple unmasking schedule. Both panels include fits of Eq.~\eqref{eq:exp_fit}, for which we only include data with $t \ge 300$. }
    \label{fig:roundsvsler}
\end{figure}
\begin{figure*}
    \centering
     \includegraphics[width=\textwidth]{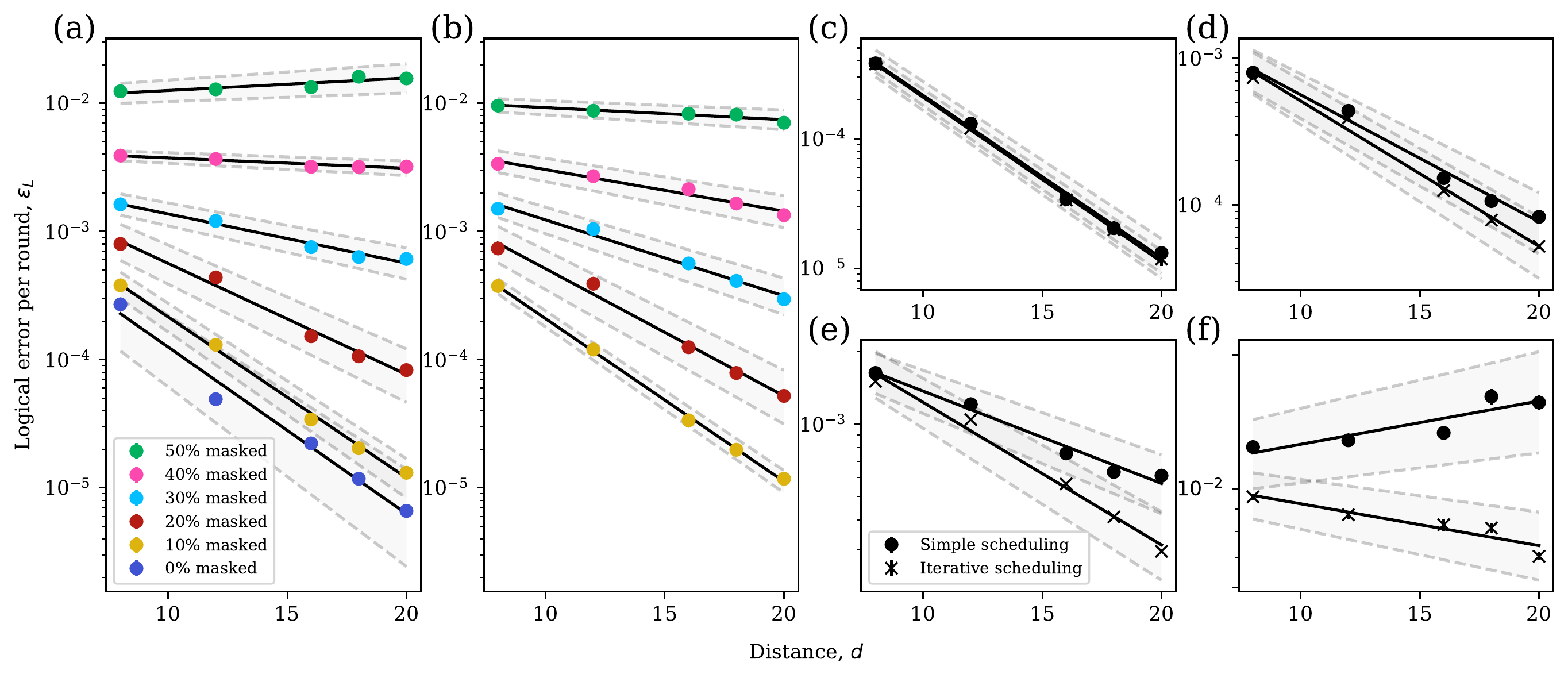}
    \caption{(a) Semilog plot of logical error rate per round, $\epsilon_L$, as a function of code distance for the simple unmasking schedule and an error rate of $p=0.001$. The fits are of a linearized Eq.~\eqref{eq:exp_error_supp} with log $\epsilon_L$. (b) Similar results for iterative scheduling. Note that we do not include 0\% masking in this case because it is equivalent to the simple schedule. Panels (c)-(f) plot the same data from panels (a)-(b) and provide easier comparisons between the simple (dot markers) and iterative (x markers) scheduling for $p_\text{mask} = \{10\%, 20\%, 30\%, 50\% \}$, respectively. The shaded region for all panels indicates error bars for $C$ and $\Lambda$. }
    \label{fig:lambda}
\end{figure*}

In this section, we report on the results of numerical simulations of a multi-round decoding protocol as described in Section~\ref{sec:masking_stacked}. Previous work has investigated the single-round performance of HGP codes using a variety of decoders \cite{Kovalev_Prabhakar_Dumer_Pryadko_2018, Grospellier_Krishna_2019, Roffe_White_Burton_Campbell_2020, Grospellier_Groues_Krishna_Leverrier_2021} and gives evidence for thresholds at near-state-of-the-art error rates. Here, we provide alternative evidence of exponential error suppression in both masked and unmasked cases following the methodology of Ref.~\cite{Chen_Satzinger_Atalaya_Korotkov_Dunsworth_Sank_Quintana_McEwen_Barends_Klimov_et_al_2021}.

We investigate a family of HGP codes constructed from a single classical expander code family and decode them using the small-set flip decoding algorithm. HGP codes---being CSS codes---can have bit- and phase-flip errors decoded independently. Furthermore, HGP codes constructed from a single base code have equivalent parity check matrices $H_X, H_Z$, and therefore, without loss of generality, we focus on the problem of decoding $X-$type errors. The specific quantum code family considered is constructed from a base (5, 6)-LDPC code family, resulting in (12, 11)-qLDPC codes. These codes have a rate of $1/61 \approx 0.016$. The results presented here correspond to a specific (un)masking \textit{schedule}, which is a potential modification of Algorithm~\ref{alg:ft_alg} and a way of specifying when and by how much to apply a mask to the syndrome. In particular, we study the following two models:
\begin{itemize}
    \item \textit{Simple scheduling}. Apply a mask $D$ with a masking percentage of $p_\text{mask}$ to use for all $\tau$ error correction rounds. After $\tau$ rounds, remove the mask completely and perform one error correction round with the fully unmasked syndrome.
    \item \textit{Iterative scheduling}. Apply a mask $D$ with masking percentage $p_\text{mask}$. After a multiple of $10^{t-1}$ rounds, for $t > 0$, remove $1-10^{-(t-1)}\%$ of the mask. For each of these instances, remove the same portion of the mask each time. On rounds $10^{t-1}+1$, all generators that were temporarily unmasked in the previous round are re-masked until another $10^{t-1}$ rounds have passed. After $\tau$ rounds, again remove the mask completely and perform one round of error correction.
\end{itemize}

In Fig.~\ref{fig:roundsvsler}(a) we show the logical error rate, $p_{\log}$, as a function of the number of rounds for a $[[3904,64,16]]$ code and the simple unmasking schedule. Data is obtained by running Algorithm~\ref{alg:ft_alg} for a fixed number of rounds with an error rate of $p = 0.001$ while varying the masking percentage, $p_{\text{mask}}$, and then recording the percentage of samples that end with a logical error. A sample is considered to end with a logical error if the final state is not equal to the initial state, up to stabilizer elements. We extract the logical error per round, $\epsilon_L$, by fitting the data to the exponential
\begin{equation}
\label{eq:exp_fit}
    p_\text{log} = 1 - (1 - \epsilon_L)^{t}.
\end{equation}
The error bars on the fits are taken from the standard error of sampling a binomial distribution, $\sqrt{p_\text{log}(1-p_\text{log}) / N}$. 

In the bottom panel of Fig.~\ref{fig:roundsvsler}, we now fix $p_\text{mask} = 0.1$ and show the performance of the simple unmasking schedule across the code family. The codes are labelled with their parameters as described in Section~\ref{subsec:qexpander}. While finding the distance of a code is generally hard, we are able to exhaustively search through the codewords of the base classical code to determine the distance of it, as well as the corresponding HGP code. Here, we observe even spacing between curves on the semilog plot showing exponential error suppression with code size. This behavior is more easily seen as the linear downwards trend in Fig.~\ref{fig:lambda}, which we now more precisely quantify.

\begin{table}
  \centering
    \begin{tabular}{||c|c|c||}
        \toprule
        $p_\text{mask}$ & Simple & Iterative \\
          & scheduling & scheduling\\
        \midrule
         0\% & $1.820 \pm~0.046$ & - \\ 
         10\% & $1.782 \pm~0.019$ & $1.794 \pm~0.010$ \\
         20\% & $1.490 \pm~0.026$ & $1.579 \pm~0.026$ \\
         30\% & $1.193 \pm~0.015$ & $1.314 \pm~0.018$ \\
         40\% & $1.038 \pm~0.007$ & $1.161 \pm~0.015$ \\
         50\% & $0.956 \pm~0.014$ & $1.044 \pm~0.009$ \\
        \bottomrule
    \end{tabular}
    \caption{Extracted values of $\Lambda$ for different masking percentages and schedules.}
    \label{table:lambdas}
\end{table}


We can relate a code family and values for logical error per round with an exponential error suppression factor $\Lambda$. For simple models, the equation
\begin{equation}
\label{eq:exp_error_supp}
    \epsilon_L = \frac{C}{\Lambda^{(d+1)/2}},
\end{equation}
where $C$ is a fitting constant and $d$ is the distance of the code, heuristically describes this relationship well.
In Fig.~\ref{fig:lambda}(a) and (b), we show the logical error per round as a function of code distance for the simple and iterative schedules, respectively. For each masking percentage, we fit a linearized Eq.~\eqref{eq:exp_error_supp} with log $\epsilon_L$ to obtain $\Lambda$. These values are listed in Table~\ref{table:lambdas}. A value of $\Lambda > 1$ is a clear indication of operating below the threshold, as increasing the code size gives an exponential decrease in the logical error rate per round. For simple scheduling, we find that for masking percentages below 50\%, $\Lambda$ is in this regime. Increasing $p_\text{mask}$ decreases $\Lambda$, and between 40\% and 50\% we see a transition where $\Lambda < 1$. In this case, it is no longer advantageous to increase the code size, as it actually causes more logical errors to occur. 

The results of the iterative unmasking schedule are shown in Fig.~\ref{fig:lambda}(b), where we find that it outperforms the simple schedule. For smaller masking percentages, it is not as advantageous to use a schedule with more unmasking, as there is less difference in performance between small masking percentages and completely unmasking (see Fig.~\ref{fig:roundsvsler}(a)). However, larger masking percentages appear to benefit more from using a more frequent unmasking schedule. In fact, with iterative scheduling, it is now the case that 50\% masked is back in the $\Lambda < 1$ regime, although with very little error suppression. Fig.~\ref{fig:lambda}(c)-(f), highlights this difference between schedules.

In both cases, we find that the results exceed the guarantees provided by Theorem~\ref{thm:ft_synd_error}. We find that the percolation threshold of this family of $(12, 10)-$qLDPC codes is around 2\%; however, we see exponential error suppression at error rates up to $\sim 50\%$, well above this threshold. 

\subsection{2D Hyperbolic Surface Codes}

As a comparison, we benchmark the performance of a 2D hyperbolic surface code on the multi-round decoding protocol. Although codes based on tilings of closed hyperbolic surfaces have a comparatively poor asymptotic distance, $d = \Theta(\log n)$, they have a constant encoding rate. These parameters violate the Bravyi-Poulin-Terhal bounds \cite{Bravyi_Poulin_Terhal_2010}, and therefore embedding these codes in 2D Euclidean space is not possible without nonlocal connections. 
However, they are in some sense close to being local, and so they make a good candidate for the stacked model.
For the construction and threshold simulations of these codes, we point the interested reader to Ref.~\cite{Breuckmann2016}. As the SSF decoder is not known to work for 2D hyperbolic surface codes, we instead use the minimum-weight perfect matching (MWPM) decoder \cite{higgott2023sparse}.
While we no longer have the guarantees of Theorem~\ref{thm:ft_synd_error}, the MWPM decoder can be modified to work with masked stabilizer generators. 
To do this, we set the nodes corresponding to masked generators as boundaries in the matching graph and set the corresponding syndrome bits to zero. Decoding normally, it is then possible to match unpaired syndrome nodes to the boundary.
Note that the standard solution to decoding with syndrome noise of building a 3D matching graph with a time dimension does not work since the mask is fixed from round to round, and the repeated measurements provide no additional information.

\begin{figure}[!t]
    \centering
    \includegraphics[width=\linewidth]{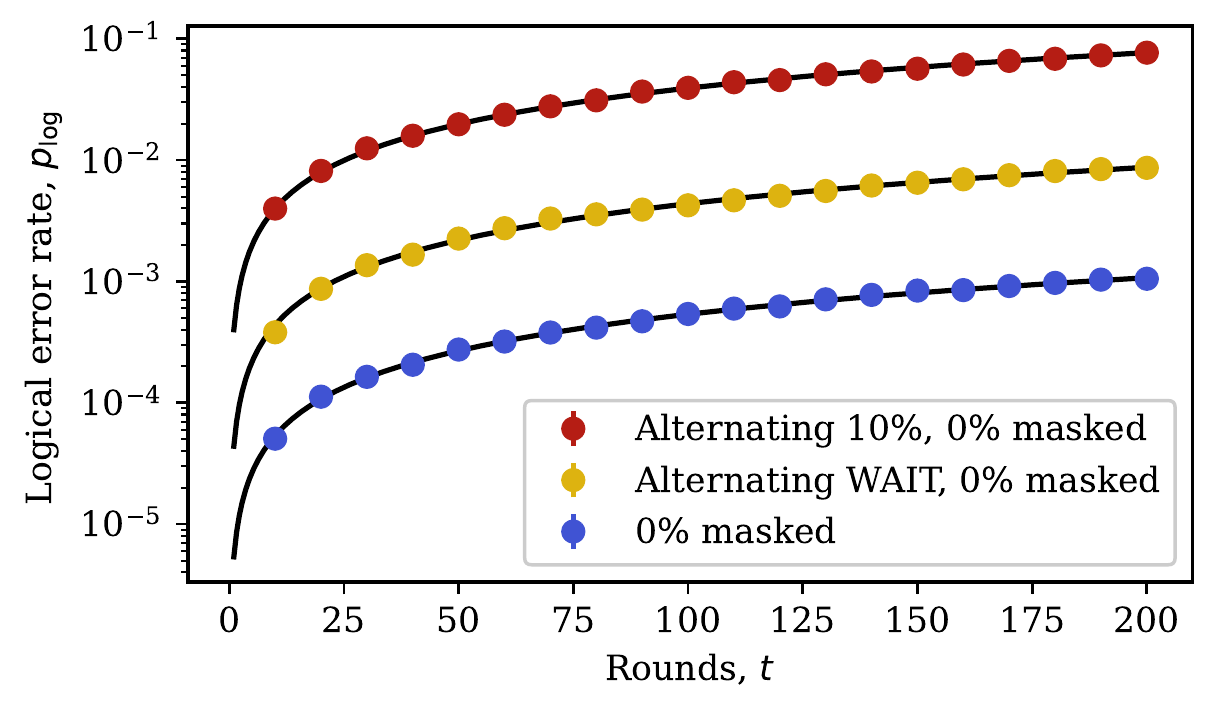}
    \caption{Semilog plot of logical error rate as a function of the number of rounds for a $[[360, 25, 8]]$ 2D hyperbolic surface code and an error rate of $p = 0.003$. We compare fully unmasked decoding performance (blue markers) with two iterative unmasking schedules. Yellow markers denote a schedule consisting of alternating between a round where no error correction is performed and a round where the entire syndrome is available. Red markers denote a schedule where masks of 10\% and 0\% are used to decode, alternating each round. Fits are of Eq.~\eqref{eq:exp_fit}. }
    \label{fig:hyperbolic}
\end{figure}

The code we investigate comes from a family of $\{5,4\}$-codes with an asymptotic rate of $1/10$ and has parameters $[[360,38,8]]$. In Fig.~\ref{fig:hyperbolic} we show the results of running Algorithm~\ref{alg:ft_alg} with this code and an error rate of $p = 0.003$ for several iterative unmasking schedules.
We compare completely unmasked decoding (blue markers) with a schedule that alternates between performing no error correction and 0\% masked each round (yellow markers) and one where the masking percentage alternates between 10\% and 0\% masked each round (red markers). 
For these codes and decoder, we find that it is actually better to do nothing and let the errors accumulate rather than try to correct the errors with the partial syndrome.
We note that we did not observe this behavior for HGP codes, even at the higher error rate.
This result is interesting as it seems to imply that the masking behavior for HPG codes is non-trivial.

One possible explanation for the difference is the single-shot~\cite{Bombin_2015, Campbell_2019} property of the SSF decoder, a property not found in the MWPM decoder.
Intuitively, this means that the syndrome has redundancies that make it more resilient to syndrome errors and masks.
Over a multi-round decoding procedure, the single-shot property also ensures that the size of any residual error is proportional to the size of the syndrome error. 
Consequently, misdiagnosing an error cannot have immediate effects throughout the system, since the size of the resulting error is bounded. 
This is not the case with the MWPM decoder, where a well-placed syndrome error could result in a long error chain across the lattice. 

\section{Discussion}
\label{sec:discussion}

In this paper, we investigated the feasibility of performing error correction with partial syndromes and found reasonable performance while masking a large constant fraction of the generators. With these results, we have motivated a new practical protocol based on the stacked model for implementing nonlocal qLDPC codes on quantum hardware restricted to 2D local gates. We note that while this limitation has been the main motivation for this work, it is possible that architectures where connectivity is not as much of a constraint might still benefit from such a protocol. Even for architectures like neutral atoms or trapped ions with effectively all-to-all connectivity, nonlocal gates are still costly in the sense that transport of the qubits is required to perform them. Limiting the number of these operations could provide overhead improvements.
There are a number of questions that need to be answered to determine whether this procedure is feasible in general.

\begin{itemize}
    \item \textit{What families of codes are amenable to the stacked model?} Theoretically, the parameters for HGP codes built from classical expander codes are allowable in this model.
    In the preparation of this work, some effort was given to find specific embeddings in $\mathbb{Z}^2$ that yielded good generator size distributions; however, the resulting distributions instead often favored mid-sized generators.
    The consequence of this is that the largest $p_\text{mask}\%$ of generators take roughly $p_\text{mask}\%$ of the work to measure. To take full advantage of the stacked model, we would instead want those largest generators to take $\gg p_\text{mask}\%$ of the work to measure.
    It is possible that other code families fit better into this model.
    One option are codes based on tessellations of closed, 4D hyperbolic manifolds~\cite{Guth2014} which are equipped with an efficient, single-shot decoder~\cite{Hastings2013DecodingIH, BreuckmannL22}. Another option are generalized bicycle codes~\cite{Kovalev2013, bravyi2023highthreshold}, which might have favorable embeddings.
    
    \item \textit{What do the syndrome extraction circuits look like for the stacked model?} The central idea of the stacked model is that the nonlocal generators are being prepared while the local generators are being used for decoding. Careful thought has to be put into the syndrome extraction circuit to ensure that we do not fall into the same pitfall of accumulating too many errors while the nonlocal generators are being prepared. A na\"ive syndrome extraction circuit consisting of SWAP gates will take $\omega(1)$ time to prepare generators of size $\omega(1)$, which is prohibitively long. Alternatively, one could use the syndrome extraction circuits of Ref.~\cite{Delfosse_Beverland_Tremblay_2021}; this method solves the scaling issue by utilizing ancilla qubits to perform long-range CNOT gates in constant depth. Remaining technicalities include the use of entanglement distillation~\cite{Bennett1996, Divencenzo1996} to ensure the resulting long-range CNOT gates are of high enough fidelity.

    \item \textit{How long does it take to perform a set of masked syndrome measurements?} As discussed in the previous question, performing the syndrome extraction of a single generator can be accomplished in constant time. However, when restricted to $O(n)$ ancilla qubits, the same cannot be said for a growing number of nonlocal generators. 
    Bounds on the depth of 2D local circuits needed to measure the \textit{full} syndrome of a stabilizer code were developed in Ref.~\cite{Delfosse_Beverland_Tremblay_2021}. 
    Extending these bounds to include specifying generator size distributions will help inform explicit unmasking schedules, which may provide better performance than the arbitrarily chosen ones studied in this work. These three questions form the basis for a practical implementation of the stacked model and are the focus of future work~\cite{Berthusen_2023}.
\end{itemize}

Several decoders for HGP codes including belief propagation~\cite{Grospellier_Groues_Krishna_Leverrier_2021} and ordered statistic decoding~\cite{Roffe_White_Burton_Campbell_2020, Panteleev_2021_osd} have been shown to perform better than the SSF decoder. An interesting question is whether these decoders work as well with the addition of masked generators.
Further improvements to the simulation can be gained by using a more realistic fault-tolerance model; in general, the error correction itself can be noisy and result in errors on the qubits and syndrome. Ultimately, performing noisy, circuit-level simulations of the syndrome extraction similar to those done in Ref.~\cite{Delfosse_Beverland_Tremblay_2021} will determine whether this protocol is possible as a whole.

\section*{Acknowledgements}

We thank Antoine Grospellier and Anirudh Krishna, whose code was useful for performing the simulations. D.G.~is partially supported by the National Science Foundation (RQS QLCI grant OMA-2120757).

\section*{Data Availability}
The source code and data to generate the figures in the paper are provided freely at \url{https://github.com/noahberthusen/hgp_partial_syndrome}. 

\bibliography{bibliography}

\end{document}